\providecommand{\keywords}[1]{\textbf{\textit{Keywords:}} #1}
\documentclass[12pt]{article}
\title{Stack and register complexity of radix conversions}
\author{Motoya Machida\footnote{Tennessee Tech.\ University, email: {\tt mmachida@tntech.edu}}\and 
Alexander Y.~Shibakov\footnote{Tennessee Tech.\ University, email: {\tt ashibakov@tntech.edu}, Corresponding author}}
\usepackage{amsthm}
\usepackage{amsmath}
\usepackage{amssymb}
\begin{document}
\maketitle
\begin{abstract}
We investigate the question of computational resources (such as stacks
and counters) necessary to perform radix conversions. To this end it
is shown that no PDA can compute the significand of the best
$n$-digit floating point approximation of a power of an
incommensurable radix. This extends the results of W.~Clinger. We also
prove that a two counter machine with input is capable of such conversions. On
the other hand we note a curious asymmetry with respect to the order in
which the digits are input by showing that a two counter machine can
decode its input online if the digits are presented in the
most-to-least significant order while no such machine can decode its input in this
manner if the digits are presented in the least-to-most significant
order. Some structural results about two counter machines (with input) are also
established.
\end{abstract}
\newtheorem{lemma}{Lemma}
\newtheorem{corollary}{Corollary}
\newtheorem{proposition}{Proposition}
\newtheorem{theorem}{Theorem}
\newtheorem{definition}{Definition}
\newtheorem{question}{Question}

\def\ttd{{\tt d}}
\def\ttD{{\tt D}}
\def\ttC{{\tt C}}
\def\ttb{{\tt b}}
\def\dg#1{\mathop{\to}\limits^{#1}}
\def\df#1{\mathop{\Rightarrow}\limits^{#1}}
\def\ndf#1{\mathop{\Rightarrow}\limits^{#1}\varnothing}
\def\nndf#1{\mathop{\not\Rightarrow}\limits^{#1}}
\def\dgn#1{\mathop{\to}\limits^{d_{#1}}}
\def\gstop{\diamond}

\section{Introduction}
Among D.~Matula's pioneering papers that laid the foundation of
modern floating-point arithmetic is \cite{Matulainout}, that investigates the
subject of radix conversions. A number of different authors produced a
variety of results dealing with the efficiency and precision issues of
conversions of floating-point numbers between different radices. The
two papers that formed the basis for subsequent work in this area are
\cite{Steele} and \cite{Clinger}. 

While the subject of principal concern for most authors working in
the field of computer arithmetic is the efficiency (both time and
space) of the algorithms performing such conversions, the question of
minimal `resources' needed for such computations was raised already by
D.~Matula in \cite{Matulainout} and further investigated by W.~Clinger
in \cite{Clinger}.

W.~Clinger's results in \cite{Clinger} have been used to justify the
use of infinite precision 
arithmetic by all known algorithms dealing with radix
conversions. One of his lemmas states that there
is no finite automaton that consumes a string of digits representing
an exponent and outputs the first digit of the best approximation of
the corresponding power of some {\ttD} in a radix that is not
commensurable with {\ttD}. He provides a separate proof for each direction of the input 
(i.e.~least or most significant digit first) and then points
out that his proof of the former is
somewhat incomplete in the sense that it does not work for all possible radix
combinations, although it does succeed for the most common case of base 10
being used for the exponent encoding and 2 for the new radix. Somewhat less important,
the proof only deals with case of the precision
of the converted result being $\geq 4$ and not $\geq 2$ as would seem
intuitively sufficient.

\iffalse
W.~Clinger's results can be restated by borrowing the language of automatic
sequences (see \cite{Allouche} for a reference) as follows.
Given positive integers {\tt D} and {\tt d} let $S_n(\ttd, \ttD)$ be
defined so that for some $k$, $S_n(\ttd, \ttD)\cdot d^k$ is the best
1-digit (for simplicity, assume ${\tt d}>2$, and pick rounding to even
to settle ambiguous cases) approximation of
$\ttD^n$. Then Lemma~8 of \cite{Clinger} implies that 
$\{\,S_n(\ttd, \ttD)\mid n\in{\mathbb N}\,\}$ is not $\ttb$-automatic for any integer
$\ttb$. Using the robustness of $\ttb$-automaticity (see
\cite{Allouche}, Theorem~5.2.3), one can see that the case of the
least-significant digit first exponent input now follows from the general
properties of $\ttb$-automaticity, and does not require a separate
number theoretic argument such as Lemma~10 of \cite{Clinger}. 
\fi

W.~Clinger's results can be restated by borrowing the language of automatic
sequences (see \cite{Allouche} for a reference) as follows.
Given positive integers {\tt D} and {\tt d} let $S_N(\ttd, \ttD)$ be
defined so that for some $k$, $S_N(\ttd, \ttD)\times \ttd^k$ is the best
1-digit (pick rounding to even to settle ambiguous cases)
approximation of $\ttD^N$ if ${\tt d}>2$;
otherwise, it represents the best 2-digit approximation.
Then Lemma~8 of \cite{Clinger} implies that 
$\{\,S_N(\ttd, \ttD)\mid N\in{\mathbb N}\,\}$ is not $\ttb$-automatic for any integer
$\ttb$. Using the robustness of $\ttb$-automaticity (see
\cite{Allouche}, Theorem~5.2.3), one can see that the case of the
least-significant digit first exponent input now follows from the general
properties of $\ttb$-automaticity, and does not require a separate
number theoretic argument such as Lemma~10 of \cite{Clinger}. 

The simple argument above shows that W.~Clinger's results indeed imply
that a fixed amount of memory is not enough to implement basic radix
conversions. From the computability perspective
though, it is still interesting to investigate how complex radix
conversion algorithms must be. Here, we use {\it complex\/} in a naive
sense, as a measure of the sophistication of the computational
`machinery' involved in implementing an algorithm. To be somewhat more
precise, having seen that a finite automaton cannot perform the
computations we require, one can ask whether a push-down automaton is
enough. An automaton with two stacks? Two counters? Questions like
these have been posed before. As an example, see~\cite{holkutr}, whose authors
investigate the register complexity of programs composed of various
looping constructs.
 
A push-down automaton (PDA for short, see \cite{hum}) can produce output
using a function that decodes its final state
in a manner similar to a deterministic finite automaton with output (DFAO,
see \cite{Allouche}).

In this paper we aim at establishing a `computability boundary' for
the task of radix conversions in the sense just outlined. We show that
the addition of a stack is not enough to carry out the required
computations.

In the second part of the paper we turn our attention to automata with
two {\it counters\/} (i.e.~two stacks whose stack alphabets consist of
a single symbol). We define a {\it two counter Minsky machine with
  input\/} (TCMI) by analogy with DFAO and provide a simple proof that such a
machine can compute any radix conversions. The subject of two counter
machines has a long history dating back to the original paper by
M.~Minsky~\cite{Minsky}. A curious phenomenon was noted early on (see
\cite{Schroep}, \cite{Bardz}, and \cite{ibarra}) that the full power of
a two counter machine can only be `tapped' via an exponential encoding
of its input and output. Without such an encoding, even the simplest
functions like $n^2$ are not computable by a TCM (see \cite{Schroep},
\cite{Bardz}, \cite{ibarra} for this and many other results of the same
flavor).

The addition of an input brings about a new level of complexity
because an encoding is supplied automatically. As an
example, to the best of the authors' knowledge, it is still unknown
whether a TCM can compute $n$ when given $2^n$ in one of the counters
(see \cite{Schroep}). On the other hand, a TCMI can simply count the
zeros in its input to output $n$ when the input is the digits of
$2^n$ in radix 2 (the subject of using a different radix is a separate
problem). 

We show (see~Proposition~\ref{mosttoleast}) that whenever the digits of $n$
(in an arbitrary positive integer radix $\ttb$) are input starting
with the most significant digit, a TCMI can compute the value of $n$
{\it online\/}, i.e.~the value of $n$ is available in one of the
counters as soon as the input is stopped (i.e.~no {\tt stop} marker is
necessary). Somewhat surprisingly, it can be shown that no TCMI is
capable of such a feat if the digits of $n$ are input starting with the
{\it least\/} significant one (see Theorem~\ref{tcexact}). It is
unknown to the authors whether a TCMI can compute the value of $n$ if
a {\tt stop} marker is a part of the input alphabet along with the
radix $\ttb$ digits. We also present some evidence that a TCMI may in
some sense be
computationally {\it weaker\/} than TCM when the input to the TCMI is 
presented least significant digit first. Namely we show (see
Theorem~\ref{mmach} and Theorem~\ref{tcmio}) that for an
unbounded function $f$ computable by a TCM (in the sense that such a
TCM halts with the value of $f(n)$ after having been `loaded' $n$ in
one of its counters) a TCMI that outputs $f(n)$ upon being
presented the digits of $n$ in least to most significant order exists
if and only if a
TCMI exists that can compute $n$ using the same input (in which case a
TCMI can obviously compute $f(n)$). This
asymmetry with respect to the order in which the input is presented is
rather unexpected in light of the result about the robustness of
automatic sequences mentioned above and the intuitive perception that
a DFAO has a very limited `memory' compared to a TCMI.  

Before proceeding with the formal definitions and statements of the
main results of this work, it is instructive to take another
look at the arguments in \cite{Clinger}. The core of W.~Clinger's proof
is formed by his Lemma~9 (see \cite{Clinger} or Lemma~\ref{ecli} below
for a slightly weaker statement) 
and {\it Kronecker's lemma\/} each highlighting different aspects of
the dynamic behavior of irrational numbers such as $\log_\ttd\ttD$ for
an {\it incommensurable\/} pair $(\ttd, \ttD)$ (the irrationality of
$\log_\ttd\ttD$ can be taken as the definition of
incommensurability of such a pair). 

Such dynamic behavior manifests
itself in many areas of science and mathematics. For some interesting
connections to other areas of mathematics see, for example,
\cite{Arnold}, Ch.~3, Exercise~4 that illustrates a curious relationship
between Poincar\'e's recurrence theorem (similar in spirit to
Kronecker's lemma) and the digits of powers
of~2. Also closely related to this subject is {\it Benford's law\/} of digit 
distribution (see \cite{Benford} and
\cite{Newcomb}, as well as \cite{Hill}; \cite{Rauchetal} provides a
curious application of Benford's law to economic forensics) which
emphasizes a statistical facet of $\log$ type dynamics.

On the other hand, a simple but clever Lemma~9 of \cite{Clinger}
quickly leads to deep number theoretic questions such as the normality
and automaticity of $\log 2$ and similar numbers, automaticity of the digit sequence
of $\sqrt 2$, etc.\ (see \cite{Allouche}) if one wishes to obtain a
stronger inequality.
 
\section{Basic definitions and notation}\label{autom}
To provide some motivation for the results in this section let us begin
by restating the problem of {\it 
  computing\/} a conversion as a problem of {\it recognizing\/} the
digits of the result of the conversions. Lemma~10 of \cite{Clinger},
declared `redundant' above takes on unexpected significance as it
seems that the part of the proof in \cite{Clinger} based on
Kronecker's lemma (see \cite{Hardy}) does not lend itself to a similar
generalization. 

Below we use the notation $a\{p\}$, where $p\in \mathbb N$, and $a$ is a
letter in some alphabet $\Gamma$, to mean a string in $\Gamma^*$ that
is a concatenation of $p$ copies of $a$.

Following the established tradition, we also use $\{\theta\}$ to
denote the fractional part of $\theta$
(i.e.~$\{\theta\}=\theta\bmod1$). This should not cause any confusion
with the use of $a\{k\}$ as a regular expression for a string of
$a$'s. Note that $\{x+y\}$ is continuous at
every $(x,y)$ such that $(x+y)\bmod 1\not=0$.

Given the `input radix' {\tt D}, the `output radix' {\tt d}, and the
`exponent radix' $\ttb$, consider the following language $P=\{\,
{\tt 10}\{p\}\mid p\in {\mathbb N}\,\}$ where {\tt 1} and {\tt 0} are
$\ttb$-digits.
Let $n \ge 1$ be fixed.
This language can be partitioned into $P_{m,n}$,
${\tt d}^{n-1} \le m <{\tt d}^n$, where 
$$
P_{m,n}=\{\,{\tt 10}\{p\}\mid m\times {\tt d}^k\hbox{ is the
best $n$-digit approximation of ${\tt D}^{\ttb^p}$}\,\}
$$
We again assume that $n \ge 2$ if $\ttd = 2$.
It is easy to see that if
one of $P_{m,n}$ is not a regular language,
then the sequence $S_N({\tt D},{\tt d})$
above is not $\ttb$-automatic thus proving that radix
conversions cannot be computed using finite automata. For some
combinations of $\ttb$ and {\tt d} we have the following stronger statement:
\begin{lemma}\label{cling}
Suppose that ${\tt d}$ and ${\tt D}$ are incommensurable.
Provided $\ttb^2/(\ttb-1)<2{\tt
    d}^{n-1}\log{\tt d}$, some of $P_{m,n}$ are not context-free.
\end{lemma}
\begin{proof}Suppose $P_{m,n}$ is context free for each
${\tt d}^{n-1} \le m <{\tt d}^n$.
Then Lemma~\ref{cycle} implies that for
each such $m$ there
exist $p,q\in \mathbb N$ such that for any $k\in\mathbb N$ ${\tt
  10}\{p+kq\}\in P_{m,n}$. Putting $Q$ to be the product of all $q$'s
  we conclude that for a large enough $p$, if ${\tt 10}\{p\}\in P_{m,n}$
  then so is ${\tt 10}\{p+Q\}$. Using Lemma~9 of \cite{Clinger} there
  is an arbitrarily large $k\in\mathbb N$ such that
$$
{\frac{\ttb-1}{\ttb^2}}<
\left\{\ttb^k(\ttb^{p+Q}-\ttb^p)\log_{\tt d}{\tt D}\right\}
<{\frac{\ttb^2-\ttb+1}{\ttb^2}}
$$
The same argument as that of Lemma~10 of \cite{Clinger} shows that
this contradicts ${\tt 10}\{p+k\}$ and ${\tt 10}\{p+k+Q\}$ both being
in $P_{m,n}$.
\end{proof}

A slightly surprising feature of the proof above is its dependence on
a particular relationship between the different 
radices. Intuitively, no such dependence should exist. The
trivial nature of the languages $P_{m,n}$ also suggests that the
inequality of Lemma~9 of \cite{Clinger} could be improved if more had been
known about the distribution of $\ttb$-digits of $\log_\ttd\ttD$ for
different radices $\ttb$. It seems, however, that even the most basic
questions of this kind (such as, how often, if at all, a certain
digit appears in the decimal expansion of $\log_\ttd\ttD$) are rather
hard (see~\cite{Allouche} or \cite{Hardy} for some examples).

We use standard definitions for most concepts appearing in this paper
as well as their natural extensions. If
$z$ is a string of $\ttb$-digits for some radix $\ttb$, we write
$\ulcorner z\urcorner$ to indicate the value of the corresponding
number in radix $\ttb$ where we assume that the least significant
digit of $z$ is the leftmost one. If the least significant digit of $z$
is the rightmost one we let $\llcorner z\lrcorner$ stand for the value
of $z$ in radix $\ttb$.

In the proofs below we only use the natural correspondence between
PDAs and CFLs (see \cite{hum}) and thus do not need the definition of
a PDA. Since  a {\it deterministic\/} PDA provides a good introduction
to two counter machines treated later let us define this narrower
concept.

A {\it deterministic push-down automaton\/} (see \cite{hum}, we assume
for simplicity that the stack is changed one symbol at a time) $M$ is defined as a 7-tuple $(Q, \Sigma,
\Gamma, q_0, Z_0, F, \delta)$ where $Q$ is a finite set of
{\it states}, $\Sigma$ is the {\it input alphabet}, $\Gamma$ is the {\it stack alphabet},
$q_o\in Q$ is the {\it initial state}, $Z_0$ is the {\it start
  symbol\/}, $F\subseteq Q$ is a set of {\it final states\/}, and the
partial function
$\delta:Q\times(\Sigma\cup\{\epsilon\})\times\Gamma\to Q\times\Gamma$
is a collection of {\it moves\/}. Some additional restrictions are
placed on $\delta$: each move $\delta(q,a,x)=(p,\cdot)$ is either a {\it pop
  move\/} (where $\cdot=\epsilon$), i.e.~``if the input symbol is $a$, the top stack symbol is
$x$ and the current state is $q$, remove $x$ from the stack and go to
state $p$'', a {\it
  push move\/} with a similar natural meaning, namely, ``under
the circumstances as above, push $\cdot$ on the stack'', or a {\it no change
  move\/}. If $a=\epsilon$, the interpretation of
$\delta(q,\epsilon,x)$ is ``ignore the input for the moment, do
something to the stack and go to
the next state $\ldots$''. We require that whenever $\delta(q,\epsilon,x)$ is
defined, no other $\delta(q,\cdot,x)$ is defined (i.e.~the DPDA is
never asked to choose whether or not to consume the input). Such moves
are called {\it $\epsilon$-moves\/}
and can be thought of as the post- or preprocessing performed by the
DPDA. Naturally, when the stack is empty, no pop moves are
possible. Empty stack can be recognized when the special symbol $Z_0$
is on top of the stack. $Z_0$ cannot be popped or pushed.

The input alphabet is $\Sigma=\{\,0,1,\ldots,\ttb-1,\gstop\,\}$
everywhere below, where $\ttb$ is some fixed radix.
The existence of $\epsilon$-moves is the reason the input alphabet
includes a {\it stop
  marker\/}, $\gstop$, to give $M$ one more chance at processing the
stack it has accumulated. We will assume that $\gstop$ means the input
is finished and that it appears only once at the end of the input. This
is not part of a standard DPDA definition but is assumed everywhere
below. If the sequence of $\epsilon$-moves following
the appearance of $\gstop$ in the input does not affect the stack, we
say that $M$ processes its input {\it online} (this concept becomes
much more important for two counter machines with input defined later).

%Pairs $(q,w)$, where $q\in Q$ and $w\in\Gamma^*$ will represent a
%specific {\it configuration\/} of $M$. The stack grows from left
%to right in this notation. 

One can view $M$ as a machine (which is, indeed, the terminology often
used in this context) that starts in $q_0$ with only $Z_0$ on the
stack then reads and processes its input one
symbol at a time, until it sees $\gstop$ (in the general case this is
unnecessarily restrictive but we will always follow this convention) upon
which it enters the final phase of processing consisting of some
$\epsilon$-moves until it ends up in one of the states in $F$. The
particular state of $M$ at the end of the computation is $M$'s {\it
output\/} and can be thought of as a finite encoding of the result $M$
is built to produce.

\iffalse
Following the established tradition, we also use $\{\theta\}$ to
denote the fractional part of $\theta$
(i.e.~$\{\theta\}=\theta\bmod1$). This should not cause any confusion
with the use of $d\{k\}$ as a regular expression for a string of
$d$'s. Note that $\{x+y\}$ is continuous at
every $(x,y)$ such that $(x+y)\bmod 1\not=0$.
\fi

\section{Radix conversions and PDAs}

To show that PDAs cannot compute radix conversions, we again restate
the computation problem as a recognition problem for the languages
defined below.
\begin{definition}Let $\ttb$, $\ttd$, and $\ttD$ be fixed radices. Let
  $n=1$ if $\ttd>2$ and $n=2$ if $\ttd=2$. Define $L_d$
  to consist of all sequences of $\ttb$-digits $z$ such that the best $n$-digit
  approximation in radix $\ttd$ of $\ttD^{\ulcorner z\urcorner}$ is
  $d$ if $\ttd>2$ or ${\tt 1}d$ if $\ttd=2$. $M_d$ can be defined similarly with $\ttD^{\llcorner
    z\lrcorner}$ instead.
\end{definition}
Just as before, it is immediate that if one of the $L_d$'s or $M_d$'s
is not context-free there is no PDA that computes the best $n$-digit
floating point approximation of $\ttD^e$ where $e$ is presented in
radix $\ttb$ in the appropriate order.

%%%%%%%%%%%%%%%%%%%%%%%%%%%%%%%%%%%%%%%%%%%%%%%%%%%%%%%%%%%%%%%%%%%%%%%

To show that one of the $L_d$'s as well as one of the $M_d$'s is not
context-free we modify the standard pumping lemma (see e.g., Theorem~7.18 of~\cite{hum})
to pump without disrupting the prefix and the suffix. The result of
the lemma is also a corollary of a very powerful {\it Multiple Pumping
  Lemma\/} (see \cite{Kracht}, Theorem~1.82). The proof of the
Multiple Pumping Lemma has not been published, however, so we present
the following direct proof instead.

\begin{lemma}\label{cycle}
Suppose that $L$ is a CFL.
Let $a\in\Sigma$ and $u,w\in\Sigma^*$ be fixed.
Then there exists a number $N$ such that
if $ua\{n_1\}w\in L$ with $n_1 \ge N$ then
we can find $1 \le \Delta \le N$ satisfying
$ua\{n_1+i\Delta\}w\in L$ for $i \ge -1$.
\end{lemma}
\begin{proof}
If $|uw| = 0$
then the claim is obvious from the standard pumping lemma.
Thus, we assume $|uw| \ge 1$.
The proof follows that of Theorem~7.18
of~\cite{hum} almost exactly, 
using a Chomsky normal form (CNF) grammar which expresses $L$.

Let the CNF grammar have $m$ variables,
and set $m' = (|uw|+1)m$ and $N = 2^{m'}$.
Suppose that $ua\{n_1\}w\in L$ with $n_1 \ge N$.
Then the longest path of a parsing binary tree for $ua\{n_1\}w$
has $(m''+2)$ edges with $m'' \ge m'$,
containing $m''$ productions
of the form $A_{i} \to A_{i+1}B_{i+1}$
or $A_{i} \to B_{i+1}A_{i+1}$ for $i=1,\ldots,m''$,
starting from the root $A_1$.
We can find 
a list of consecutive variables
$A_{i_1},A_{i_1+1},\ldots,A_{i_2}$ of size at least $m$
such that the subtree $B_{i+1}$ generates a substring of $a\{n_1\}$
for every $i_1 \le i \le i_2$,
and consequently we can find $A_{j_1} = A_{j_2}$
such that $i_2-m+1 \le j_1 < j_2 \le i_2+1$.

We write the yields of the subtrees $A_{j_1}$
and $A_{j_2}$ as $xvy$ and $v$ respectively.
As in the proof of the standard pumping lemma we have
$|xvy| \le N$ and $u'x\{i\}vy\{i\}w'\in L$ for any $i \ge 0$,
where $ua\{n_1\}w = u'xvyw'$.

Furthermore, one of the following cases must hold for $v$:
({\it i\/}) $v = a\{k\}$ with $k \ge 1$;
({\it ii\/}) $v = u''a\{k'\}$ with suffix $u''$ of $u$ and $k' \ge 0$;
or
({\it iii\/}) $v = a\{k'\}w''$ with prefix $w''$ of $w$ and $k' \ge 0$.
Now we can find an appropriate $\Delta \ge 1$
as required. Namely,
if ({\it i\/}) holds then we put $xy = a\{\Delta\}$;
if ({\it ii\/}) holds then $x = \epsilon$ so put $y = a\{\Delta\}$;
if ({\it iii\/}) holds put $x = a\{\Delta\}$ since $y = \epsilon$.
\end{proof}

Next we proceed with the number theoretic results that take advantage
of the combinatorial lemma above. The first statement is the
famous Kronecker's lemma in a slightly weaker form than the original
(see~\cite{Hardy} for the full version and the proof).

\begin{lemma}[Kronecker's lemma, see~\cite{Hardy}] The set
  $\{\,\{n\theta\}\mid n\in{\mathbb N}\,\}$ is dense in $(0,1)$ for
  every irrational $\theta>0$.
\end{lemma}

The following lemma is an easy corollary of \cite{Clinger},
Lemma~9 whose proof is based on the analysis of the fractional part of
$\theta$ presented in radix $C$.

\begin{lemma}[\cite{Clinger}]\label{ecli}Let $\theta>0$ be irrational, $C>1$ be a
  natural number. Then there exist
  infinitely many $m>0$ such that for $\alpha=\{C^m\theta\}$,
  $\beta=\{C^{m+1}\theta\}$:
$$C^{-2}<|\alpha-\beta|<1-C^{-2}$$
\end{lemma}

Unfortunately, the inequality above is too weak for our goals and has
to be amended. At present we do not have the number theoretic tools to
produce a `clean' proof of a better inequality and have to take an
indirect approach, instead, by modifying a few digits of the number. The first
modification puts the iterate of $\theta$ in one of the two ranges. It
seems that it should be possible to ensure that it ends up in a
specific range, however it is unclear how to do that at the moment.
 
\begin{lemma}\label{qrational}
Let $\theta>0$ be irrational, $C>1$ be an integer. Then there exists
an integer $K>0$ and an infinite sequence $m_1,\ldots,m_i,\ldots$ of
integers such that either 
$$
1/3\leq\lim_{i\to\infty}|\{KC^{m_i}\theta\}-\{KC^{m_i+1}\theta\}|\leq1/2
$$ or
$$
1/2\leq\lim_{i\to\infty}|\{KC^{m_i}\theta\}-\{KC^{m_i+1}\theta\}|\leq2/3
$$
\end{lemma}
\begin{proof}
Using Lemma~\ref{ecli} find an infinite sequence $m_1,\ldots,m_i,\ldots$ such that
$C^{-2}<|\{C^{m_i}\theta\}-\{C^{m_i+1}\theta\}|<1-C^{-2}$. Using the
sequential compactness of $[C^{-2},1-C^{-2}]$, and picking a
convergent subsequence if necessary, assume that
$L=\lim_{i\to\infty}|\{C^{m_i}\theta\}-\{C^{m_i+1}\theta\}|$ exists
and $C^{-2}\leq L\leq 1-C^{-2}$. If $L$ is irrational, use Kronecker's
lemma to find $K>0$ such that $1/3<\{KL\}<1/2$. Otherwise, $L=p/q$
where $p$ and $q$ are relatively prime and $q\geq2$. Let integers $m$ and
$n$ be chosen so that $np+mq=1$; then $nL=1/q-m$. Thus
$\{|n|L\}=1/q$ or $\{|n|L\}=1-1/q$ depending on the signs of $m$ and
$n$. After multiplying $|n|$ by $q-1$ if necessary, we can assume
$\{|n|L\}=1/q$. Putting $K=|n|\lfloor q/2\rfloor$, we have
$1/3\leq\{KL\}\leq 1/2$. 

Now, for each $i$, either
$|\{KC^{m_i}\theta\}-\{KC^{m_i+1}\theta\}|=\{K|\{C^{m_i}\theta\}-\{C^{m_i+1}\theta\}|\}$
or
$|\{KC^{m_i}\theta\}-\{KC^{m_i+1}\theta\}|=1-\{K|\{C^{m_i}\theta\}-\{C^{m_i+1}\theta\}|\}$. Thus,
after possibly choosing a proper subsequence, either
$\lim_{i\to\infty}|\{KC^{m_i}\theta\}-\{KC^{m_i+1}\theta\}|=\{KL\}$ or
$\lim_{i\to\infty}|\{KC^{m_i}\theta\}-\{KC^{m_i+1}\theta\}|=1-\{KL\}$.
\end{proof}

The second modification establishes the desired inequality.
\begin{lemma}\label{normalize}
Let $\theta>0$ be irrational, $\alpha_i$ and $\beta_i$, $i=1,2,\ldots$
be such that $\alpha_i,\beta_i\in(0,1)$ and either $1/3\leq\lim|\alpha_i-\beta_i|\leq1/2$ or
$1/2\leq\lim|\alpha_i-\beta_i|\leq2/3$. Then there is an increasing
subsequence $i(j)$, $j=1,2,\ldots$ and an integer $q\geq0$ such
that
$1/3\leq\lim|\{\alpha_{i(j)}+q\theta\}-\{\beta_{i(j)}+q\theta\}|\leq1/2$.
\end{lemma}
\begin{proof}
The only nontrivial case is
$1/2\leq\lim|\alpha_i-\beta_i|\leq2/3$. Picking
a subsequence if necessary, assume $\alpha_i>\beta_i$ (the other case
is similar) and
$\lim\beta_i$ exists. Using Kronecker's lemma, pick an integer $q\geq0$
so that $\{\lim\beta_i+q\theta\}>1-1/4$. Then $m-1/4<\lim\beta_i+q_2\theta<m$
for some integer $m>0$ so $m+2/3>\lim\alpha_i+q_2\theta>m+1/4$.
Put $\alpha=\{\lim\alpha_i+q\theta\}$ and
$\beta=\{\lim\beta_i+q\theta\}$.

Now $\alpha<2/3<3/4<\beta$ and
$\lim|\{\beta_i+q_2\theta\}-\{\alpha_i+q_2\theta\}|=|\lim\{\beta_i+q_2\theta\}-\lim\{\alpha_i+q_2\theta\}|=|\{\lim\beta_i+q_2\theta\}-\{\lim\alpha_i+q_2\theta\}|=\beta-\alpha$.
Thus $\lim\beta_i+q_2\theta=(m-1)+\beta$ and
$\lim\alpha_i+q_2\theta=m+\alpha$ imply
$\lim|\alpha_i-\beta_i|=\lim(\alpha_i+q_2\theta)-(\beta_i+q_2\theta)=(m+\alpha)-(m-1+\beta)=1-(\beta-\alpha)$. Hence
$1/3\leq\beta-\alpha\leq1/2$. 
\end{proof}

The lemma below deals with pairs of irrationals. It establishes a
finite bound on the number of iterations required to put each in a
desired `slot'.
\begin{lemma}\label{mixing}Let $\theta\in\mathbb R$ be irrational,
$0<a<b<1$, $\epsilon>0$. Then there
exists an $n(\theta,\epsilon,a,b)\in\mathbb N$ such that for any
$\alpha,\beta>0$ satisfying
$\{\beta\}<\{\alpha\}$ and
$b-a+\epsilon<\{\alpha\}-\{\beta\}<1-\epsilon$ there exists a $k\leq
n(\theta,\epsilon,a,b)$ such that $\{\beta+k\theta\}<a$ and
$\{\alpha+k\theta\}>b$.
\end{lemma}
\begin{proof}Kronecker's lemma implies the existence of
$n(\theta,\epsilon,a,b)$ such that for any $\tau\geq0$ the set 
$P=\{\,\{\tau+k\theta\}\mid k<n(\theta,\epsilon,a,b)\,\}$ has the property
$P\cap(\max\{0,a-\epsilon\},a)\not=\varnothing$ and 
$P\cap(\max\{b,1-\epsilon\},1)\not=\varnothing$ Now consider two cases. 

(1) $\{\alpha\}-\{\beta\}\leq1-a$. Let $k<n(\theta,\epsilon,a,b)$ be such 
that $\{\beta+k\theta\}\in(\max\{0,a-\epsilon\},a)$. Let
$\{\beta\}+k\theta=m+\{\beta+k\theta\}$ for some integer $m\geq0$.
Then $\{\alpha\}+k\theta\leq m+\{\beta+k\theta\}+(1-a)<m+a+(1-a)=m+1$.
Also
$\{\alpha\}+k\theta\geq m+\{\beta+k\theta\}+b-a+\epsilon>
m+a-\epsilon+b-a+\epsilon=m+b$. Thus $\{\alpha+k\theta\}=\{\{\alpha\}+k\theta\}=\{\alpha\}+k\theta-m>b$.

(2) $\{\alpha\}-\{\beta\}>1-a$. Let $k<n(\theta,\epsilon,a,b)$ be such that
$\{\alpha+k\theta\}$ is in $(\max\{b,1-\epsilon\},1)$ and
$\{\alpha\}+k\theta=m+\{\alpha+k\theta\}$ for some integer $m\geq0$.
Then
$\{\beta\}+k\theta\geq m+\{\alpha+k\theta\}-(1-\epsilon)>m+(1-\epsilon)-(1-\epsilon)=m$
and $\{\beta\}+k\theta\leq
m+\{\alpha+k\theta\}-(1-a)<m+1-(1-a)=m+a$. Hence
$\{\beta+k\theta\}=\{\{\beta\}+k\theta\}=\{\beta\}+k\theta-m<a$.
\end{proof}
Note that the choice of $k$ in the lemma above implies that
$\{\alpha+k\theta\}-\{\beta+k\theta\}=\{\alpha\}-\{\beta\}$.

The results above are put to use in the following theorem. Recall that
an {\it $n$-digit best floating point approximation\/} in radix $\ttd$ to a real $r$ is a
floating point number $m\times\ttd^q$ such that
$r=(m+\epsilon)\times\ttd^q$ where $\ttd^{n-1}\leq m<\ttd^n$, $|\epsilon|\leq1/2$ and where
$m=\ttd^{n-1}$ only if $-1/(2\ttd)\leq\epsilon$. Such an approximation
is unique whenever $\epsilon<1/2$ which is the only case we need below.
\begin{theorem}\label{pdacs}
If $\ttd$ and $\ttD$ are incommensurable, $\ttd>2$ or $n>1$ 
then there are $d_1,d_2<\ttd$ such that $L_{d_1}$ and
$M_{d_2}$ are not context-free.
\end{theorem}
\begin{proof}
Put $C=\ttb^\Delta$ and $\theta=\log_\ttd\ttD$ where $\ttb$ is the radix in which the exponent is
presented and $\Delta$ is as in Lemma~\ref{cycle}. Use
Lemma~\ref{qrational} to find a sequence of $m_i$'s and a $K>0$ such
that either $1/3\leq\lim_{i\to\infty}|\{KC^{m_i}\theta\}-\{KC^{m_i+1}\theta\}|\leq1/2$ or
$1/2\leq\lim_{i\to\infty}|\{KC^{m_i}\theta\}-\{KC^{m_i+1}\theta\}|\leq2/3$
holds. Now apply Lemma~\ref{normalize} and possibly rename the terms
of the sequence to find $q$ so that 
$1/3\leq\lim|\{KC^{m_i+1}\theta+q\theta\}-\{KC^{m_i}\theta+q\theta\}|\leq1/2$.

Picking another subsequence if necessary assume
$\alpha_i=\{KC^{m_i+1}\theta+q\theta\}>\beta_i=\{KC^{m_i}\theta+q\theta\}$ 
(the case of the opposite inequality is similar).
Suppose $0<a<b<1$ are chosen so that $b-a<1/3$. Omitting finitely many initial
terms of the sequence if necessary, one can find an $\epsilon>0$ such
that $1-\epsilon>\alpha_i-\beta_i>b-a+\epsilon$ for all
$i\in\mathbb N$. Apply Lemma~\ref{mixing} to find an
$n(\theta,\epsilon,a,b)$ such that for any $i$ there exists a
$k<n(\theta,\epsilon,a,b)$ with the property that 
$\{KC^{m_i}\theta+q\theta+k\theta\}<a$ and
$\{KC^{m_i+1}\theta+q\theta+k\theta\}>b$. Thinning out the
sequence again one can assume that there is an integer $p\geq0$ such
that for every $i$ $\{KC^{m_i}\theta+q\theta+p\theta\}<a$ and
$\{KC^{m_i+1}\theta+q\theta+p\theta\}>b$.
 
\iffalse %%%%%%%%%%%%%%%%%%%%%%%%%%%%%%%%%%%%%%%%%%%%%%%%%%%%%%%%%%%%%%%%%%%%%%%%%%%%%%%
Let $u$ be the ${\ttb}$-digits of $p+q$, and $w$ be
the ${\ttb}$-digits of $K$. Use Lemma~\ref{cycle} to find $n_0$ and $n_1$ then pick $m_i$
large enough so that $m_i>n_0$, and $m_i$ satisfies an additional
property mentioned below.
Lemma~\ref{cycle} implies that both 
$$
u{\tt 0}\{n_1+(m_i-n_0)\Delta\}w=\mbox{``$\ttb$-digits of $q+p$''}\mbox{\tt
  0}\{(m_i+1)\Delta-|u|\}\mbox{``$\ttb$-digits of $K$''}
$$ and 
$$
u{\tt 0}\{n_1+(m_i+1-n_0)\Delta\}w=\mbox{``$\ttb$-digits of $q+p$''}\mbox{\tt
  0}\{m_i\Delta-|u|\}\mbox{``$\ttb$-digits of $K$''}
$$ are in the same
$L_d$ (or $M_d$ after the strings
have been reversed accordingly and the remark after Lemma~\ref{cycle}
has been 
applied). Thus the $\ttd$-significands of $\ttD^{K\ttb^{(m_i+1)\Delta}+q+p}$
and $\ttD^{K\ttb^{m_i\Delta}+q+p}$ are the same. 
\fi %%%%%%%%%%%%%%%%%%%%%%%%%%%%%%%%%%%%%%%%%%%%%%%%%%%%%%%%%%%%%%%%%%%%%%%%%%%%%%%%%%

Let $u$ be the $\ttb$-digits of $p+q$, and $w$ be the $\ttb$-digits
of $K$.
Use Lemma~\ref{cycle}, and choose $n_0$ so that
$n_1 = n_0\Delta - |u| \ge N$.
Then pick $m_i$ large enough so that $m_i > n_0$, and $m_i$
satisfies an additional property mentioned below.
Lemma~\ref{cycle} implies that both
$$
u{\tt 0}\{n_1+(m_i-n_0)\Delta\}w = 
\mbox{``$\ttb$-digits of $q+p$''}\mbox{\tt
  0}\{m_i\Delta-|u|\}\mbox{``$\ttb$-digits of $K$''}
$$
and
$$
u{\tt 0}\{n_1+(m_i+1-n_0)\Delta\}w = 
\mbox{``$\ttb$-digits of $q+p$''}\mbox{\tt
  0}\{(m_i+1)\Delta-|u|\}\mbox{``$\ttb$-digits of $K$''}
$$
are in the same $L_d$
(or $M_d$ after the strings have been reversed accordingly).
Thus the $\ttd$-significands of $\ttD^{K\ttb^{(m_i+1)\Delta}+q+p}$
and $\ttD^{K\ttb^{m_i\Delta}+q+p}$ are the same. 

First suppose $\ttd>2$,
$\ttD^{K\ttb^{(m_i)\Delta}+q+p}=\sum_{j=0}^Na_j\ttd^j$, and
$\ttD^{K\ttb^{(m_i+1)\Delta}+q+p}=\sum_{j=0}^Mb_j\ttd^j$, where
$a_j,b_j<\ttd$, $a_N>0$, $b_M>0$. Now
$\{(KC^{m_i}+q+p)\theta\}=\{(K\ttb^{m_i\Delta}+q+p)\theta\}=\log_\ttd(a_N+\sum_{j=0}^{N-1}a_j\ttd^{j-N})$
therefore picking $a=\log_\ttd(1+\eta)$ with $\eta\leq1/2$ implies
that the first digit of the significand of (the best approximation of)
$\ttD^{K\ttb^{m_i\Delta}+q+p}$ is $1$. Let $b=\log_\ttd2$
and suppose it is possible to choose $\eta\leq1/2$ so that
$\log_\ttd(1+\eta)+1/2<\log_\ttd(\ttd-1/2)$. Then for some $\epsilon'>0$
$\log_\ttd(1+\eta)+1/2+\epsilon'<\log_\ttd(\ttd-1/2)$. Let $m_i$ be
chosen large enough so that 
$\{KC^{m_i+1}\theta+q\theta\}-\{KC^{m_i}\theta+q\theta\}<1/2+\epsilon'$. Now
$\log_\ttd(b_M+\sum_{j=0}^{M-1}b_j\ttd^{j-M})=\{(KC^{m_i+1}+q+p)\theta\}$
and
$\{(KC^{m_i+1}+q+p)\theta\}=\{(KC^{m_i}+q+p)\theta\}+\{(KC^{m_i+1}+q+p)\theta\}-\{(KC^{m_i}+q+p)\theta\}<
\log_\ttd(1+\eta)+\{(KC^{m_i+1}+q+p)\theta\}-\{(KC^{m_i}+q+p)\theta\}<\log_\ttd(1+\eta)+1/2+\epsilon'<\log_\ttd(\ttd-1/2)$. Thus
the first digit of the significand of
$\ttD^{K\ttb^{(m_i+1)\Delta}+q+p}$ is between $2$ and $\ttd-1$
contradicting our assumption.

The case of $\ttd=2$ is somewhat special since the first digit of the
significand is always $1$. In this case we pick
$a=\log_2(1+\eta)$ for $\eta<1/4$, $b=\log_2(3/2)$ and $\epsilon'>0$
so that $\log_2(1+\eta)+1/2+\epsilon'<\log_2(1+1/2+1/4)$. An argument
similar to the one for $\ttd>2$ shows that the second digit of the
significand of $\ttD^{K\ttb^{(m_i+1)\Delta}+q+p}$ is $0$ while the
second digit of the significand of $\ttD^{K\ttb^{m_i\Delta}+q+p}$
is~$1$.

To complete the proof it is sufficient, for every $\ttd>1$, to pick
$\eta\leq1/2$ ($\eta<1/4$ in the case $\ttd=2$) so that 
$\log_\ttd(1+\eta)+1/2<\log_\ttd(\ttd-1/2)$ and $\log_\ttd2-\log_\ttd(1+\eta)<1/3$
(respectively $\log_2(1+\eta)+1/2<\log_2(1+1/2+1/4)$ and
$\log_23/2-\log_2(1+\eta)<1/3$ for $\ttd=2$).

Consider three cases.

(1) $\ttd\geq4$. Put $\eta=1/2$. Now it can be easily
verified that $\log_42-\log_4(3/2)<1/3$. Since $\log_xr$ decreases as
$x>0$ increases for any $r>1$ it follows that
$\log_\ttd2-\log_\ttd(3/2)<1/3$. A simple computation also shows that 
$\log_\ttd(3/2)+1/2<\log_\ttd(\ttd-1/2)$ for $\ttd\geq4$.

(2) $\ttd=3$. We must pick an $\eta\leq1/2$ so that 
$\log_32-1/3<\log_3(1+\eta)<\log_3(5/2)-1/2$. That such $\eta$ exists
follows from $2\,3^{-1/3}-1<1/2$ and $5\,2^{-1}3^{-1/2}>2\,3^{-1/3}$.

(3) $\ttd=2$. We must pick $\eta\leq1/4$ so that
$\log_2(3/2)-1/3<\log_2(1+\eta)<\log_2(3/2+1/4)-1/2$. The existence of
$\eta$ is a consequence of $3\,2^{-1}2^{1/3}-1<1/4$ and
$7\,4^{-1}2^{-1/2}>3\,2^{-1}2^{-1/3}$.
\end{proof}

\section{Clinger's problem for two counter machines}

Implicit in~\cite{Clinger} is the following problem.
\begin{definition}
Consider the problem of determining whether a
specific class of machines can compute the best $n$-digit
approximation in radix $\ttd$ to $f\times\ttD^e$,
where $f$ and $e$ are integers and $f$ is positive. If a machine from a given class can
compute the significand of the best approximation we say that such a machine
``solves Clinger's problem of the floating point arithmetic''.
\end{definition}

D.~Matula (see \cite{Matulainout} and \cite{Clinger})
demonstrated that this problem can be solved by a deterministic finite
automaton (DFA) if $\ttd$ and $\ttD$ are commensurable.
By Theorem~\ref{pdacs} we now know precisely that it is the only affirmative result
for PDA (or DFA), which we state as the following corollary.
\begin{corollary}
Clinger's problem of the floating point arithmetic can be solved by a PDA
(or DFA) if and only if $\ttd$ and $\ttD$ are commensurable.
\end{corollary}

Having seen that a DPDA is too limited to
compute radix conversions, it is natural to seek a more powerful
machine to accomplish the task. A straightforward modification of the
DPDA concept immediately leads to the definition of a {\it
  deterministic push-down automaton with two stacks} (DPDA2S for short). Such automata,
however, are {\it too\/} powerful for our purposes. To limit
their computational power we can consider DPDA2S whose stack
alphabets are limited to two symbols each (two, because we need a
special $Z_0$ that indicates the bottom of each stack). It is easy to
see that the configuration of such an automaton can be described by its
current state and the two values representing the number of symbols other
than $Z_0$ currently on the stack. Hence, each stack acts simply as a
{\it counter\/} and the construction just presented defines a {\it two
counter machine with input\/} or TCMI for short.

A TCMI is capable of carrying out very sophisticated
computations even when no input is present. Indeed, when the
only moves a TCMI is allowed are $\epsilon$-moves, we arrive at the
concept of a {\it two register Minsky\/} machine (TCM). As M.~Minsky showed
in \cite{Minsky}, for any recursive function $f$ there is a TCM that
computes $f(n)$ in the sense that a computation that starts with $2^n$
in one register (counter) and zero in the other terminates with the
value of $2^{f(n)}$ in one of the registers and zero in the
other (\cite{Minsky} used $3^{2^n}$ to encode the output but $2^n$ is
enough, see \cite{Bardz}). Another result in \cite{Minsky} shows that the addition of yet
another, third register would allow such a machine to compute $f(n)$
directly, starting with $n$ (rather than $2^n$).

A number of authors had proved that the addition of a third register is
indeed necessary, and that the exponential encoding cannot be bypassed if
one wishes to retain the full computational power of a TCM (see
\cite{Schroep}, \cite{ibarra}, \cite{Bardz}, and \cite{holkutr}). The authors of
\cite{ibarra} prove that even a {\it characteristic function\/} of the
set of exact squares, for example, is not computable by a TCM. Neither
are elementary functions such as $2^n$, $n^2$, or $\lfloor\log n\rfloor$ (see
\cite{Bardz} and \cite{Schroep}).

A TCM can be thought of as a program in a simple language
$\{\,x\leftarrow x+1$, $x\leftarrow x-1$, {\bf if $x=0$ then goto
  $l$}, {\bf goto $l$}, $\mbox{\bf halt}\,\}$ where $x$ can be one of
the two variables. A TCMI extends this language to include a number of `{\bf if
  $I=d$ goto $l$}' commands, one for each possible value $d$ of the
special input variable $I$. It is convenient to think of a TCMI as a
collection of several TCM's sharing the two counters, each TCM
dedicated to the processing of a specific input symbol. If such a TCM
that processes the stop marker, $\gstop$ never changes the counters, we say that the
TCMI processes its input {\it online} by analogy with the DPDA case.
After each TCM is done with the processing it arrives at a line (which
we will think of as a state, $w$) with one of the input commands. We
will refer to such a line as a {\it wait state\/}.

%%%%%%%%%%%%%%%%%%%%%%%%%%%%

We can show that a TCMI is capable of performing the task
of radix conversions, and present a brief sketch of the proof of the
following proposition.
\begin{proposition}\label{tcmo}
For every combination of radices $\ttd$, $\ttD$, and
  $\ttb$, and any precision $n$
Clinger's problem of the floating point arithmetic can be solved
by a TCMI online.
\end{proposition}
\begin{proof}
The methods of \cite{Minsky} show that a TCM can simultaneously
compute the values of {\it several\/} recursive functions
$f_1,f_2,\ldots$ by keeping its input (and output) encoded in the form
$2^{n_1}3^{n_2}5^{n_3}\ldots$. It should be noted that the original
arguments of \cite{Minsky} used a more sophisticated encoding, namely $3^{2^n}$ for the output but
several authors (see e.g.~\cite{Bardz}) had observed that $p^n$ used
for the input as well as $p^{f(n)}$ for the output where $p$ is some prime would
suffice. 

Thus if $n_1$ holds the value of the
number input `so far', $n_2$ holds the number of digits seen (to take
into account any leading zeros in the case of a least significant
digit first input), and $n_3$ represents the significand
of the best $n$-digit approximation using some encoding, the
computation can proceed by decoding the value of $5^{n_3}$ using
successive divisions (of which there could only be a bounded number as
there are only finitely many values a significand can assume).

Now, upon seeing the next digit the TCMI simply increments
$n_2$ (multiplying the counter by $3$), and computes the new value of
$n_1$ using a straightforward computation followed by the computation
of $n_3$ which is obviously a recursive function of $n_1$.
\end{proof}

\section{Counting and the asymmetry of input}

It is still unknown (at least to the authors) whether a TCM is capable
of {\it decoding\/} its own output, i.e., for example to halt with the
value of $k$ after having started with $2^k$ in one of the registers.
It {\it is} known that a TCM is incapable of computing $\lfloor
\log_2k\rfloor$ (see \cite{Schroep}) but as R.~Schroeppel points out
in \cite{Schroep} these problems are not equivalent since a TCM that
computes $k$ from $2^k$ is not even assumed to halt on any input other than
$2^k$.

It is thus natural to ask whether a TCMI exists that decodes its
{\it input\/}, i.e.~terminates with the value of $n$ after being input the
digits of $n$ followed by $\gstop$. This problem can be thought of as a
conversion to radix 1 if desired. If the nature of the input is
restricted, such decoders are certainly possible (see
Proposition~\ref{mosttoleast} below). 
%Moreover, in contrast to
%a TCM, if the digits of $2^n$ are input in radix 2, all a TCM has to
%do is count the number of {\tt 0}'s in the input to output $n$.

Let a TCMI $M$ compute the value of $n$ after being input the
digits of $n$ in some radix $\ttb$ followed by $\gstop$. If the digits
are assumed to be presented starting with the least significant one,
we say that $M$ {\it counts in radix $\ttb$\/} or simply {\it counts
  in $\ttb$\/}. If the digits are presented starting with the most
significant one we say that $M$ {\it counts in $\ttb$ in reverse}.
We have the following simple statement. Recall that computing {\it
  online\/} means that upon the input of $\gstop$ the values of the
registers do not change.
\begin{proposition}\label{mosttoleast}
For any radix $\ttb$ there exists a TCMI that counts in $\ttb$ online
in reverse.
\end{proposition}
\begin{proof}
When it is input $d$ the TCMI multiplies the value it has computed so
far by $\ttb$ and adds $d$ to the result.
\end{proof}

The authors of \cite{ibarra} introduce a number of tools that simplify
dealing with TCM's one of which is the idea of a {\it normalized\/}
TCM (called nTCM below). By requiring that each individual TCM in a TCMI be
normalized we arrive at the definition of a normalized TCMI or
nTCMI. While this concept is not required to carry out the proofs
below, the handling of some border cases is simplified when a TCMI is
actually an nTCMI so we tacitly assume that each TCMI is normalized
(we only need this concept in Lemma~\ref{linper} below). Since one of
the results in \cite{ibarra} is the fact that each TCM can be replaced
by an equivalent nTCM, this is not a significant restriction, and we
thus omit the definition of nTCM  and nTCMI here.

To simplify the notation, we will often record a configuration of a
TCMI $M$ as $(q^k,n^k_1,n^k_2)$ where $q^k$ is the state, $n^k_1$ and $n^k_2$
are the values of the counters, and $k$ is an integer index. Using this notation
define $n^k_-=\min\{n^k_1,n^k_2\}$ and $n^k_+=\max\{n^k_1, n^k_2\}$. It is
often important to know which counter holds the smaller (larger)
value. We denote the index of the corresponding
counter as ${}_-n^k=\min\{\,i\mid n^k_-=n^k_i\,\}$
(${}_+n^k=\max\{\,i\mid n^k_+=n^k_i\,\}$ respectively). Sometimes we
know only that one of the values is the value of the first counter,
therefore the other value is the contents of the other, such as $n_u$
and $n_{3-u}$ for $u\in\{1,2\}$. In such cases we write
$(q,n_u\,n_{3-u})$, omitting the comma. Thus $(q, n_u\, n_{3-u})=(q,
n_1, n_2)$ regardless of what $u$ is.

To facilitate the handling of arguments involving computations performed by
$M$ we introduce the notation $(w,n_1,n_2)\dg{d}(w',m_1,m_2)$ meaning
$M$ will reach the configuration (although not necessarily halt) $(w',m_1,m_2)$ if started in
$(w,n_1,n_2)$ and being input $d$. If $w=q_0$ and $n_1=n_2=0$ we write
$M\dg{d}\ldots$ instead.
 
Finally, each computation is analyzed in terms of {\it stages\/} or
{\it phases\/} such that during each phase the values of both counters
are nonempty. To simplify the terminology we introduce the notation
$(w,n_1,n_2)\df{I}(w',m_1,m_2)$ to mean that the first time $M$ enters
a configuration $(w'',m'_1,m'_2)$ such that $m'_-=0$ after having
started in $(w,n_1,n_2)$ and while being input $I$ ($I$ can be a
string so such a configuration can be encountered before $I$ is fully
consumed), this configuration is $(w'',m'_1,m'_2)=(w'',m_1,m_2)$. If
having started in $(w,n_1,n_2)$ and being input a string $I$ the machine never reaches a
configuration with one of the counters empty, we write $(w,n_1,n_2)\ndf{I}$. 

The analysis of a TCMI differs from that of a TCM in one important
aspect: a TCM is usually started with at least one counter empty and
terminates in a similar state. Simple arguments show that a TCM
performing any useful computation can be restricted to operating in
this manner (see Lemmas~\ref{fstg} and \ref{lstg} below for one
reason this is so). A TCMI is a different matter though. Imagine a
TCMI that counts the number of {\tt 0}'s in its input in one counter and the number
of {\tt 1}'s in the other, then multiplies the first value by the
last digit seen (to make the example nontrivial assume that the radix
is $>2$) in the input and adds the two results. 
Even though the final output of the computation
is a single value, it is 
unclear if the same result can be achieved by a TCMI that keeps only
one of the counters nonempty before the next digit is input. 

If in all such cases
one of the values is bounded by a single constant this value can be
kept in the final control `buffer' instead, while emptying the
corresponding counter. The advantage of this would be the availability 
of various reduction results for TCMs such as the one in
\cite{Schroep} (see the proof of Theorem~\ref{tcexact} below)
for the analysis of the TCMI.

Our first goal is to show that such a bound exists {\it for every TCMI that
counts in some radix $\ttb$} thus showing that the anomaly
described above does not occur in such TCMIs.

The following lemma is a corollary of Lemma~2.2 of \cite{ibarra} (see
also \cite{ibarra}, Lemma~2.3 and the definition of an MP1RM in \cite{Schroep} and below).
\begin{lemma}\label{linper}
Let $s$ be the number of states of some nTCM $M$. Let $(q^0, n^0_1,
n^0_2)$, $\ldots$, $(q^k, n^k_1, n^k_2)$ list all the stages of some computation
performed by $M$. Suppose $n^i_-=0$ and $n^i_+>s$ for all $i\leq k$. 
Then there exist integer $P$, $Q$, $R$, and $D$ such that for any
$m\geq0$ and any computation $(q^0, m^0_1, m^0_2)$, $\ldots$, $(p^k, m^k_1,
m^k_2)$ such that $m^0_-=0$, ${}_-m^0={}_-n^0$, and $m^0_+=n^0_++mD$ $m^k_+=P\lfloor
m^0_+/Q\rfloor+R$ and $p^i=q^i$ for $i\leq k$.
\end{lemma}

The next statement describes what happens when both counters hold
large enough values. As for most results of this kind, its proof is
based on cycle analysis of the machine. 
\begin{lemma}\label{fstg}
Let $s$ be the number of states of some TCM $M$. Let $(q^0, n^0_1,
n^0_2)$, $n^0_->s+1$ be a configuration of $M$ with the following
properties. $M$ eventually reaches a waiting state when starting in $(q^0, n^0_1,
n^0_2)$. There exists a $k\geq0$ such that
$(q^0, n^0_1, n^0_2)$, $\ldots$, $(q^k, n^k_1, n^k_2)$ lists an initial stage of the computation
performed by $M$, where $n^j_->0$ for $j<k$, and
either $n^k_+>n^k_-=0$ or $M$ halts in $q^k$ with $n^k_->0$. Then every configuration $(q^0, m^0_1,
m^0_2)$ with $m^0_i\geq n^0_i$, $i=1,2$ has the same property and one of
the following two cases holds.
\begin{itemize}
\item[{\rm(1)}] There exist $r_1$, $r_2$, with the property that $|r_i|\leq s$ and
  $m_i^k=m_i^0+r_i$ for any computation $(q^0, m^0_1, m^0_2)$,
  $\ldots$, $(p^k, m^k_1, m^k_2)$ such that $m^0_i\geq n^0_i$,
  $i=1,2$. Moreover $p^j=q^j$, $k\leq s$, $m^j_->0$ for $j\leq k$,
  and $M$ halts in $q^k$.

\item[{\rm(2)}] There exist integers $\omega_1$, $\omega_2$
  with the following properties. Suppose $m^0_i\geq n^0_i$, $i=1,2$. Then 
  $(q^0, m^0_1, m^0_2)\df{}(p',m^\nu_1,m^\nu_2)$ and 
  the following additional properties hold.
\begin{itemize}
\item[{\rm(a)}] If $\omega_1\omega_2\leq0$ then $\omega_u\geq0$ and
  $\omega_{3-u}<0$ for some $u=1,2$. Moreover there exist an $r$ and a state $q'$ such that
  $|r|<3s$, and whenever $m^0_u>s+1$, $m^0_{3-u}=n^0_{3-u}+m\omega_{3-u}$, the state
  $p'=q'$ and $m^\nu_+=m^0_u+\omega_u\lfloor m_{3-u}/|\omega_{3-u}|\rfloor+r$.

\item[{\rm(b)}] 
  If $\omega_1\omega_2>0$ then both $\omega_1,\omega_2<0$ and
  there exist $r_1$, $r_2$, $|r_1|,|r_2|<3s$, and states ${}^1q$ and ${}^2q$ 
  with the following properties. Suppose $m^0_i=n^0_i+l_i|\omega_i|$ for
  some $l_i>0$, $i=1,2$. Put $u={}_+m^\nu$. Then
  $m^\nu_+=m^0_u+\omega_u\lfloor m_{3-u}/|\omega_{3-u}|\rfloor+r_u$ and $p'={}^uq$.
  Moreover if $m^0_u$ is replaced by a larger value, the length of the
  computation and the final state will not change and the expression above will remain
  valid. There also exists a constant $D$ such that $l_i>l_{3-i}+D$
  implies $i=u$. Finally, if $l_1=l_2$ then $p'=q^k$, $u$ is
  independent of $l_i$, and $m^\nu_+=n^0_u+\omega_u\lfloor n_{3-u}/|\omega_{3-u}|\rfloor+r_u$.

\end{itemize}
\end{itemize}
\end{lemma}
\begin{proof}
Suppose (1) does not hold. Then $M$ does not reach a halting state
before one of the counters is empty. Since $n^0_->s+1$, $k>s+1$. 
Therefore $M$ enters
a loop and passes through the following sequence of states (we are
assuming $M$ must halt eventually):
$q^0\ldots q^{\omega_0}\sigma^0\ldots\sigma^{t\omega-1}q^{\omega_0+t\omega+1}\ldots
q^k$,
where $\omega_0<s$, $k-t\omega-\omega_0<\omega\leq s$, and $(\sigma^{l+t'\omega},
n^{l+t'\omega}_1,
n^{l+t'\omega}_2)=(\sigma^l,n^l_1+t'\omega_1,n^l_2+t'\omega_2)$ for some $\omega_1$ and
$\omega_2$ such that $|\omega_1|,|\omega_2|\leq\omega\leq s$. 
Thus $M$ will reach the configuration $(q^k, n^k_1, n^k_2)$ with $n^k_-=0$
and $n_-^j>0$ for $j<k$ for some $k>0$.
 
Suppose $\omega_1\omega_2>0$. Note then both
$\omega_1,\omega_2<0$. Otherwise the counters will be incremented
indefinitely and the loop will not terminate. Thus for any configuration $(q^0, m^0_1,
m^0_2)$ with $m^0_i\geq n^0_i$, $i=1,2$ there is a state $p'$ such
that $(q^0, m^0_1,m^0_2)\df{}(p',m^\nu_1,m^\nu_2)$ with $m^\nu_-=0$. 

To prove (b) assume $m^0_i=n^0_i+l_i|\omega_i|$ for some $l_i>0$, $i=1,2$, and
$u={}_+m^\nu$. Let ${}^uq$ be $p'$. As long as counter~$u$
stays nonempty throughout the whole computation, the value of $l_u$ does not affect
the final state. Therefore, $p'$ depends only on the value of
counter~$3-u$ at the beginning of the last (possibly incomplete) iteration of the loop.
But this value is easily seen to be the same for any $l_{3-u}>0$ under the
conditions listed above. It remains to pick $r_u$. Let $u_i$, $i=1,2$ be the
amount counter~$i$ is incremented by before $M$ enters the loop and let
$r_{3-u}'$ be the value of counter~$3-u$ at the beginning of the last
(possibly incomplete) iteration of the loop. Note that $|u_i|<s$ and $r_{3-u}'<s$
do not depend on the values of $l_{3-u}$ and $l_u$ as long as the
conditions above hold (i.e.\ counter~$3-u$ is emptied first). Now
$m^\nu_u=m^0_u+u_u+\omega_u(m^0_{3-u}+u_{3-u}-r_{3-u}')/|\omega_{3-u}|=m^0_u+\omega_u\lfloor
m^0_{3-u}/|\omega_{3-u}|\rfloor+r_u$ for some $r_u$ independent of $l_i$ and
such that $|r_u|<3s$. Note that making the value of $m^0_{3-u}$ larger at
the beginning of the computation does not affect the argument above.

To find $D$ simply note that it is sufficient to pick $D$ large enough
to ensure that counter~$i$ is nonempty after the loop has run $\lfloor
(m_{3-i}+s)/|\omega_{3-i}|\rfloor+1$ times.

If $l_1=l_2$ then $m^\nu_+=n^0_u+l_u|\omega_u|-|\omega_u|\lfloor
(n^0_{3-u}+l_{3-u}|\omega_{3-u}|/|\omega_{3-u}|\rfloor+r_u=
n^0_u+l_u|\omega_u|-|\omega_u|l_{3-u}|\omega_{3-u}|/|\omega_u|-|\omega_u|\lfloor
n^0_{3-u}/|\omega_{3-u}|\rfloor+r_u=n^0_u+\omega_u\lfloor
n^0_{3-u}/|\omega_{3-u}|+r_u$.
That $p'=q^k$ follows from a simple observation that the values of
both counters before the final iteration of the loop are the same
whether $M$ started in $(q^0, n^0_1, n^0_2)$ or $(q^0, m^0_1, m^0_2)$.
Thus (b) holds.

Now suppose $\omega_1\omega_2<0$. Then, say $\omega_1<0$ and
$\omega_2>0$. Thus counter~1 will be emptied first as long as
$m^0_2>s+1$ (to ensure that the value of counter~2 cannot become 0
during the first iteration of the loop). Now the argument similar to
that for case~(b) above finishes the proof of case~(a).

If $\omega_1=0$ then $\omega_2\geq0$ would result in an infinite
loop. Therefore $\omega_2<0$ and counter~2 is emptied first. The rest
of the argument is similar to the one for case~(b) using the remark in
the previous paragraph.
\end{proof}

If a TCM goes through a configuration with an empty counter, the final
smaller value of the counter will be bounded. This is the statement of
the following lemma. 
\begin{lemma}\label{lstg}
Let $s$ be the number of states of some TCM $M$. Let $(q^0, n^0_1,
n^0_2)$, $\ldots$, $(q^k, n^k_1, n^k_2)$ list a nontrivial last stage of some computation
performed by $M$ (i.e.\ $n^0_-=0$, $n^i_->0$ for all $0<i\leq k$ and
$M$ halts in $q^k$). Then $n^k_-<s+1$ and any computation $(q^0, m^0_1, m^0_2)$, $\ldots$, $(p^k, m^k_1,
m^k_2)$ such that $m^0_+\geq 2s+1$, ${}_-m^0={}_-n^0$, and $m^0_-=0$
satisfies $m^k_-=n^k_-$ and $p^i=q^i$ for $i\leq k$. 
\end{lemma}
\begin{proof}
First observe that the length of the last stage must be less than
$s+1$ in order for $M$ to halt in $q^k$ (otherwise $M$ would have entered a
loop which either does not terminate or empties one of the
counters). Thus $n^k_-<s+1$ and both $n^k_+$ and $m^k_+$ are the contents of the counter that was
nonzero at the beginning of the appropriate computation (i.e.\ the
${}_+m^0$($={}_+n^0$)th
counter, we also use $m^0_+>2s+1$ here). Since the sequence of
states $M$ passes through must be identical in both cases, the other
counter gets incremented exactly the same number of times in either
computation.
\end{proof}

\begin{corollary}\label{zstg}
Let $s$ be the number of states of some TCM $M$ and $d$ be an input symbol. Suppose
$(w,n_1,n_2)\df{d}(w',n'_1,n'_2)$ and
$(w,n_1,n_2)\dg{d}(w'',m_1,m_2)$. Then $m_+<s+1$.
\end{corollary}

To ensure the applicability of Lemma~\ref{linper} the following
statement is used.
\begin{lemma}\label{lprefix}
For any $K$ there exists a constant $N_K>0$ such that any configuration $(q,n_1,n_2)$
that $M$ passes through after inputting a string of digits longer than
$N_K$ satisfies $n_+>K$.
\end{lemma}
\begin{proof}
Otherwise $M$ would output the same value for some strings of
different lengths, a contradiction.
\end{proof}

We can now turn our attention to the existence of the bound mentioned above.
\def\sgn{\mathop{{\rm sgn}}}
\begin{lemma}\label{boundt}
Suppose there exists a TCMI $M$ that counts in some radix $\ttb$. Let
$\{\,w_i\mid i<\xi\,\}$ list all the waiting states of $M$. Then
there is a constant $K>0$ such that for any input
$(w_0,0,0)\dgn0(w_{i(1)},n^1_1,n^1_2)\dgn1\,\ldots\,\dgn{\nu-1}(w_{i(\nu)},n^\nu_1,n^\nu_2)$
the bound $n^j_-<K$ holds for $i=1,2$ and $j\leq\nu$.
\end{lemma}

Before proceeding with the proof of Lemma~\ref{boundt} some
preliminary properties need to be established.
\begin{lemma}\label{unper}
Suppose, using the notation of Lemma~\ref{boundt} there is no $K$ with
the claimed property. Then there exist a digit sequence $\Omega$, 
a state $w$, a constant $s+1<v\leq(s+1)^2$, an index $u\in\{1,2\}$,
and $R_u>0$ and $R_{3-u}<0$, $|R_1|,|R_2|<(s+1)^2$, with the following
properties. 
\begin{itemize}
\item[(1)] for any $K$ there exists a digit sequence $I$ such that
$(w_0,0,0)\dg{I}(w,n_u\,n_{3-u})$ where $n_u=v$,
$n_{3-u}>K$;

\item[(2)]
$(w,m_1,m_2)\dg{\Omega}(w,m_u+R_u\,m_{3-u}+R_{3-u})$ provided $m_u\geq
v$, $m_{3-u}>(s+1)^2$.

\end{itemize}
\end{lemma}

\begin{proof}
Let $J$ be an input such that $M\dg{J}(w',n'_1,n'_2)$ and
$n'_->\max\{K,(s+1)^2\}$. If follows from Lemma~\ref{zstg} that there
exists a string $I'$ such that $I'\leq J$, $M\dg{I'}(w'',n''_1,
n''_2)$ so that $n''_-<s+1$ and $(w'',n''_1,n''_2)\ndf{I''}$ 
for any $I''$ such that $I'I''\leq J$. Let $J=Id_1d_2\ldots d_k$ and 
$(w'',n^0_1=n''_1,n^0_2=n''_2)\dgn{1}(w^1,n^1_1,n^1_2)
\dgn{2}\discretionary{\ldots}{}{}\ldots\dgn{k}(w^k=w',n^k_1=n'_1,n^k_2=n'_2)$ for some
$k$. 

Since $(w^j,n^j_1, n^j_2)\ndf{d_{j+1}}$ one can apply Lemma~\ref{fstg}(1) to
find $r_1^j$ and $r_2^j$ such that $|r_i^j|\leq s$ and
$(w^j,m^j_1,m^j_2)\dg{d_j}(w^{j+1}, m^j_1+r_1^j, m^j_2+r_2^j)$
whenever $m^j_i>n^j_i$. An easy inductive argument shows that
$w^0\ldots w^k=\Lambda_1\Lambda_2\ldots\Lambda_l$,
where $l\leq s$ and $\Lambda_j=w(j)\ldots w(j)$ or
$\Lambda_j=w(j)$. Let $u\in\{1,2\}$ be such that $n''_u=n''_-$. For
  $\Lambda_j=w^{j'}\ldots w^{j''}$, $j<l$ put 
$R^j_u=\sum_{m=0}^{j''-j'}r_u^{j'+m}$ and
$R^l_u=\sum_{m=0}^{j''-j'-1}r_u^{j'+m}$. Then
$\sum_{m=1}^lR^m_u>(s+1)^2-(s+1)=s^2+s$, therefore there exists an
$m$ such that $R^m_u\geq s+1$. Since $|r_u^j|\leq s$ the sequence of
states $\Lambda_m=\Lambda w^{j''}$ for some
$\Lambda=w^{j'}\ldots$ such that $w^{j'}=w^{j''}=w$. Let $j'$ be the
smallest such that $m$ has the desired properties. Put
$\Omega=d_{j'+1}\ldots d_{j''}$,
$R_i=\sum_{j=j'}^{j''-1}r_i^j=R^m_i-r_i^{j''}$ if $m<l$ and
$R_i=R^m_i$ if $m=l$, where $i=1,2$. It follows from $|r^j_i|<s+1$
that $R_u\geq0$. Put $v=n^{j'}_-$. 

One can assume that the length of $\Lambda$ is less than $s+1$ using
the following construction. Suppose $\Lambda$ is longer than $s+1$
and assume for the moment that $K$ is large enough as explained
below. Among the first $s+1$ states in $\Lambda$ find one, say $q'$ that
repeats at least twice. Let $R'_u$ be the `contribution' to the value
of counter~$u$ by the
subsequence of $\Lambda$ that starts with $q'$ and ends just before
$(q',\ldots)$ is reached again. If $R'_u\geq0$ pick the corresponding
input digits as $\Omega$ and redefine $R_u=R'_u$. Otherwise shorten
$\Lambda$ by removing the subsequence. This will only increase $R_u$,
and, if $n''_{3-u}$ is large enough will not cause counter $3-u$ to
become 0 (counter $u$ will only get larger than $n^{j'''}_u$ where $w^{j'''}$ is
the first state after the subsequence above). Now $|R_i|<(s+1)^2$.  
An argument similar to that of the proof of $R^m_u\geq
s+1$ shows that $s+1\leq v<(s+1)^2$.

That $(w,m_1,m_2)\dg{\Omega}(w,m_u+R_u\,m_{3-u}+R_{3-u})$ for any $m_u\geq
v$, $m_{3-u}>(s+1)^2$ now follows from $|r^j_i|\leq s$. Since $\Omega$
is shorter than $s+1$, counter $3-u$ will not become empty while the
digits of $\Omega$ are input, whereas counter $u$ values will only
increase if $\Lambda$ is trimmed as discussed above. Hence the
sequence of states $M$ follows in $(w,m_1,m_2)\dg{\Omega}\ldots$ is the
same as long as $m_1$ and $m_2$ are large enough.

We have shown that for any $K$ there are $\Omega_K$ of length less than $s+1$, $w_K$,
$s+1<v_K\leq(s+1)^2$, $u(K)\in\{1,2\}$, $R_{u(K)}\geq0$ and $R_{3-u(K)}$,
$|R_1^K|,|R_2^K|<(s+1)^2$ that satisfy property~(2) in the statement of
the lemma, as well as an $I_K$ such that $(w_0,0,0)\dg{I_K}(w,n_u\,n_{3-u})$ where $n_u=v$,
$n_{3-u}>K$. Pick $\Omega$, $w$, $v$, $u$, $R_u$, and $R_{3-u}$ such
that $(\Omega,w,v,u,R_u,R_{3-u})=(\Omega_K,w_K,v_K,u(K),R_{u(K)},R_{3-u(K)})$
for arbitrarily large $K$. For any $K$ let $I(K)$ be an input such that 
$M\dg{I(K)}(w,n_u\,n_{3-u})$ where $n_u=v$,
$n_{3-u}>K$.

To show that $R_{3-u}<0$ suppose $R_{3-u}\geq0$. Let
$I=I((s+1)^2)$. Then $M\dg{I(K)}(w,n_u\,n_{3-u})$ where $n_u=v$,
$n_{3-u}>(s+1)^2$.  Consider what happens after the following input:
$I\Omega\{l\}{\tt 1}\gstop$. In the trivial case of $R_u=R_{3-u}=0$
inputting multiple copies of $\Omega$ does not change the
configuration $M$ enters after inputting $I$. Thus $M$ produces the
same value regardless of $l$, a contradiction. Hence either $R_u>0$ or
$R_{3-u}>0$.  If $(w,n_1,n_2)\ndf{{\tt 1}}$, since $n_->s+1$,
Lemma~\ref{fstg}(1) shows that $(w,m_1,m_2)\dg{{\tt 1}}(w',m_1+C_1,
m_2+C_2)$ for some small $C_1$ and $C_2$ whose value does not depend
on $m_1$ and $m_2$ as long as $m_i>s+1$. Since
$(w,n_1,n_2)\df{I\gstop}(q',n'_1,n'_2)$ for some $q'$, $n'_-=0$ the
proof below will remain essentially the same if we assume that
$(w,n_1,n_2)\df{{\tt 1}}(q',n'_1,n'_2)$. In this case
Lemma~\ref{fstg}(2) applies since otherwise $n_->s+1$ would imply
$n'_->0$. Pick $\omega_1$ and $\omega_2$ as in Lemma~\ref{fstg}, note
that $(w,n_1,n_2)\dg{(\Omega\{z\})\{l\}}(w,n_1+zR_1l, n_2+zR_2l)$ for
arbitrarily large $l$ and pick and integer $z$ such that $|\omega_i|$
divides $zR_i$.  Let $\Theta=\Omega\{z\}$, $Z_i=zR_i$, $i=1,2$, set
$M\dg{I}(w,n_1,n_1)\dg{\Theta\{l\}}(w,n_1+Z_1l,
n_2+Z_2l)\df{{\tt1}}(q',n'_1,n'_2)$ and consider the following cases,
depending on $w$.

First suppose $\omega_1\omega_2\leq0$. Suppose $\omega_i\geq0$ and use
Lemma~\ref{fstg}, property~(a) to conclude that $n'_{3-i}=0$,
$n'_i=n_i+lZ_i+\omega_i\lfloor(n_{3-i}+lZ_{3-i})/|\omega_{3-i}|\rfloor+r
=n_i+lZ_i+\omega_il(Z_{3-i}/|\omega_{3-i}|)+\omega_i\lfloor
n_{3-i}/|\omega_{3-i}|\rfloor+r=
l(Z_i+\omega_i(Z_{3-i}/|\omega_{3-i}|))+B=Al+B$ where $A$ and $B$
depend only on $n_1$, $n_2$, and $w$. The state $q'$ does not depend
on the value of $l$.

Second suppose $\omega_1\omega_2>0$ and use (b) of
Lemma~\ref{fstg}. If $Z_u/|\omega_u|=Z_{3-u}/|\omega_{3-u}|$ then
${}_+n'$ does not depend on $l$ (see below) so we can assume
${}_+n'=u$. Then Lemma~\ref{fstg} implies that
$n'_u=n_u+lZ_u-|\omega_u|\lfloor(n_{3-u}+lZ_{3-u})/|\omega_{3-u}|\rfloor+r
=n_u+lZ_u-|\omega_u|l(Z_{3-u}/|\omega_{3-u}|)+\omega_u\lfloor
n_{3-u}/|\omega_{3-u}|\rfloor+r=0l+B$ where $B$ does not depend on
$l$. Since $n_u+lZ_u=n_u+|\omega_u|l(Z_u/|\omega_u|)=n_u+m|\omega_u|$
and
$n_{3-u}+lZ_{3-u}=n_{3-u}+|\omega_{3-u}|l(Z_{3-u}/|\omega_{3-u}|)=n_{3-u}+m|\omega_{3-u}|$
the state $q'$ or ${}_+n'$ does not depend on $l$ according to the
remark at the end of the statement for case~(b) of Lemma~\ref{fstg}.
 
Finally suppose $Z_i/|\omega_i|>Z_{3-i}/|\omega_{3-i}|$ for some
$i\in\{1,2\}$. If $l$ is large enough then
$n_i+lZ_i=n_i+|\omega_i|l(Z_i/|\omega_i|)=n_i+al|\omega_i|$ and
$n_{3-i}+lZ_{3-i}=n_{3-i}+|\omega_{3-i}|l(Z_{3-i}/|\omega_{3-i}|)=n_{3-i}+bl|\omega_{3-i}|$
with $(a-b)l>D$ where $D$ is given by Lemma~\ref{fstg}(1)(b).  Thus
${}_+n'=i$ for all such $l$ and
$n'_i=n_i+lZ_i+\omega_i\lfloor(n_{3-i}+lZ_{3-i})/|\omega_{3-i}|\rfloor+r
=n_i+lZ_i+\omega_il(Z_{3-i}/|\omega_{3-i}|)+\omega_i\lfloor
n_{3-i}/|\omega_{3-i}|\rfloor+r=
l(Z_i+\omega_i(Z_{3-i}/|\omega_{3-i}|)+B=Al+B$ where $A$ and $B$
depend only on $n_1$, $n_2$, and $w$.

Thus in all cases $M$ arrives at the configuration
$(q',n'_1,n'_2)=(q', (Al+B)\, 0)$ where $A$ and $B$ are independent of
$l$ and ${}_+n'=i_w$ stays the same for all $l$ larger than some
constant. If $M$ is allowed to continue then following the input of
$\gstop$ it will produce the value of $\displaystyle \ulcorner
I\urcorner+\ttb^m(\ulcorner\Theta\urcorner\,{\ttb^{l\omega}-1\over\ttb-1}
+\ttb^{l\omega})$ where $m$ is the largest integer such that
$\ttb^{m-1}\leq \ulcorner I'\urcorner$ and $\omega$ is the length of
$\Theta$. It is now a simple matter to construct a TCM $M'$ that will
compute the function $f(l)=\ttb^{l\omega}$ contradicting Theorem~3.4
of \cite{ibarra}. $M'$ will first test $l$ to see if it is large
enough to satisfy $(a-b)l>D$. If not, $M'$ will use its finite control
to compute $f(l)$. Otherwise, it will compute the value of $Al+B$ and
place it in the appropriate counter. It will then pass the computation
on to $M$ by putting $Al+B$ in counter $i_w$ and emptying the other
counter. After $M$ is finished, $M'$ will (after inputting $\gstop$
and waiting for $M$ to finish again) compute $\ttb^{l\omega}$ from the
result produced by $M$ which requires only a finite number of
multiplications, additions and divisions by fixed constants. Thus
$R_{3-u}<0$. In this case $R_u=0$ implies that for large enough $l$'s
$n'_+$ is bounded, whereas $M$'s output is unbounded for long enough
$I$'s, a contradiction.
\end{proof}

We can now proceed with the proof of Lemma~\ref{boundt}.

\begin{proof}[Proof of Lemma~\ref{boundt}]
Suppose the contrary and use Lemma~\ref{unper} to find $\Omega$, $w$,
$s+1<v\leq(s+1)^2$, $u\in\{1,2\}$, $R_u>0$ and $R_{3-u}<0$, with the
appropriate properties.  Let $K$ be large enough. Consider what
happens when the input $I\Omega\{m\}$ for small enough $m\geq0$ is
followed by ${\tt 1}\gstop$ (the conditions on $K$ and $m$ are
discussed below).  Just as in the proof of Lemma~\ref{unper}, let
$M\dg{I}(w,n_1,n_1)\dg{\Theta\{l\}}(w,n_1+Z_1l,
n_2+Z_2l)\df{{\tt1}}(q',n'_1,n'_2)$ for appropriately chosen $\Theta$,
$Z_1$ and $Z_2$, let and $\omega_1$ and $\omega_2$ be the constants
provided by Lemma~\ref{fstg}(2) for the last stage of this
computation. Note that these constants are the same for all
configurations $(w,n_1+Z_1l,n_2+Z_2l)$ such that
$n_1+Z_1l,n_2+Z_2l>s+1$.

We assume, as in Lemma~\ref{unper} that both $Z_i$ is a multiple of
$\omega_i$, $i\in\{1,2\}$ (if $\omega_i$ is 0, the appropriate
requirement is considered satisfied automatically). Now consider the
following three cases.

(1) $\omega_{3-u}\geq0$ and $\omega_u<0$. Then
$(w',n'_1,n'_2)=(w',(n'_u=0)\,(n_{3-u}+\omega_{3-u}\lfloor
v/|\omega_u|\rfloor+r))=(w',(n'_u=0)\,(n_{3-u}+B))$ where both $w'$
and $B$ depend only on the value of $v<(s+1)^2$. Moreover,
$(w,n_u+lZ_u\,n_{3-u}+lZ_{3-u})\df{{\tt
    1}}(w',(n'_u=0)\,(n_{3-u}+lZ_{3-u}+l\omega_{3-u}(Z_u/|\omega_u|)+\omega_{3-u}\lfloor
n_u/|\omega_u|\rfloor+r))=(w',(n'_u=0)\,(n_{3-u}+lA+B))$ where the
state $w'$, and the constants $A$ and $B$ depend only on the value of
$v<(s+1)^2$, provided $\lfloor (v+Z_ul)/|\omega_u|\rfloor<\lfloor
(K+Z_{3-u}l-s-1)/|\omega_{3-u}|\rfloor$ (if $\omega_{3-u}=0$ this
inequality is automatically satisfied).

(2) $\omega_{3-u}<0$ and $\omega_u\geq0$. Then
$(w,n_u+lZ_u\,n_{3-u}+lZ_{3-u})\df{{\tt
    1}}(w',n'_1,n'_2)=(w',(v+lZ_u+\omega_u\lfloor
(n_{3-u}+lZ_{3-u})/|\omega_{3-u}|\rfloor+r)\,(n'_{3-u}=0))=(w',(
v+lZ_u+l\omega_u(Z_{3-u}/|\omega_{3-u}|)+\omega_u\lfloor
n_{3-u}/|\omega_{3-u}|\rfloor+r)\,(n'_{3-u}0)=(w',
(v+lA+\omega_um+r)\,(n'_{3-u}=0))$ where $A$ does not depend on $n_1$
or $n_2$, $m=\lfloor n_{3-u}/|\omega_2|\rfloor$, and $r$ depends only
on $(n_{3-u}\bmod |\omega_2|)$ which is bounded by $(s+1)$.

(3) $\omega_u<0$ and $\omega_{3-u}<0$. Using (b) of Lemma~\ref{fstg}
choosing $l$ and $k$ so that $\lfloor
(v+Z_ul)/|\omega_u|\rfloor<\lfloor
(K+Z_{3-u}l-s-1)/|\omega_{3-u}|\rfloor$ one can show that
$(w,n_u+lZ_u\,n_{3-u}+lZ_{3-u})\df{{\tt
    1}}(w',n'_1,n'_2)=(w',(n'_u=0)\,n_{3-u}+lZ_{3-u}+\omega_{3-u}\lfloor(v+lZ_u)/|\omega_u|\rfloor+r)=(w',(n'_u=0)\,
(n_{3-u}+lZ_{3-u}+l\omega_{3-u}(Z_u/|\omega_u|)+\omega_{3-u}\lfloor
v/|\omega_u|\rfloor+r)=(w', v+lA+B\,0)$ where $A$ does not depend on
$n_1$ or $n_2$, $B$ depends only on $v$.

Construct an infinite sequence of strings $\{\,I_n\mid n\in{\mathbb
  N}\,\}$ satisfying the following properties. There are $b<s+1$, and
$c$ such that for each $j$ $M\dg{I_j}(w,v\,n_{3-u}(j))$ where
$n_{3-u}(j)>j$, either $\omega_{3-u}=0$ or the remainder
$b=n_{3-u}\bmod|\omega_{3-u}|$, and $c=n_{3-u}(j) \bmod
|A\omega_{3-u}|$ provided $|A\omega_{3-u}|$ is not 0. Note that all
$(w,v\,n_{3-u}(j))$ satisfy the same property (1)--(3) mentioned
above.

Note that the conditions above and the cases~(1)--(3) discussed before
the previous paragraph imply for some state $w'$
$(w,v\,n_{3-u}(j))\df{{\tt 1}\gstop}(w', m_1, m_2)$.

Using Lemma~\ref{lprefix} to ensure that all the configurations $M$
enters after $M\dg{I_j}(w,v\,n_{3-u}(j))$ have at least one counter
larger than $s+1$ and taking a large enough $j$ consider the output of
$M$ given the input $I_j{\tt 1}\gstop$. Using Lemma~\ref{linper} find
the constants $P$, $Q$, $R$, and $D$ with the property that the output
of $M$ is $P\lfloor m_+(j)/Q\rfloor +R$ where $(w',m_1(j),m_2(j))$ is
the first configuration $M$ enters after $I_j{\tt 1}$ has been input
such that $m_-(j)=0$.

By taking a multiple of $D$ we can assume that every appropriate
constant mentioned below divides $D$. Next revisit the three cases
above.

(1) Find $k$ large enough so that one can find $l$ and $l+l'$ as
described below, where $l'>2$ that are both solutions of
$n_{3-u}(k)+lA+B=n_{3-u}(j)+B+mD$ for some $m$.  To see that this is
possible note that $\omega_{3-u}\not=0$ so either $A=0$ (which is
impossible) or $c=n_{3-u}(k)\bmod |A\omega_{3-u}|=n_{3-u}(j)\bmod
|A\omega_{3-u}|$.  Let $D=D'|\omega_{3-u}A|$ and let
$n_{3-u}(k)=n_{3-u}(j)+m'|\omega_{3-u}A|+m''D$ where $m'<D'<D$. Now
put $l=(D'-(\sgn A)m')|\omega_{3-u}|$, $l+l'=(2D'-(\sgn
A)m')|w_{3-u}|$. It is easy to see that $l,l+l'\leq 3D|\omega_{3-u}|$
so the existence of a large enough $k$ so that the conditions in
case~(1) above are satisfied is immediate.  Then, on the one hand, $M$
should output $P\lfloor (n_{3-u}(k)+lA+B)/Q\rfloor+R$ using the choice
of $P$, $Q$, $R$, and $D$. One the other hand, $M$ will output
$\displaystyle \ulcorner I_k\Omega\{l\}{\tt 1}\urcorner=UE^l+V$ where $\displaystyle
U={\ttb^e(\ulcorner\Omega\urcorner+\ttb^f-1)\over \ttb^f-1}$,
$\displaystyle V=\ulcorner I_k\urcorner-{\ulcorner\Omega\urcorner
  \ttb^e\over \ttb^f-1}$, $E=\ttb^f$, and $e$ is the length of $I_k$
and $f$ is the length of $\Omega$. Thus $P\lfloor
(n_2(k)+lA+B)/Q\rfloor+R=UE^l+V$ and $P\lfloor
(n_k(k)+(l+l')A+B)/Q\rfloor+R=UE^{l+l'}+V$ but $(P\lfloor
(n_k(k)+(l+l')A+B)/Q\rfloor+R)-(P\lfloor (n_2(k)+lA+B)/Q\rfloor+R)\leq
l'PA/Q+1$ while $(UE^{l+l'}+V)-(UE^l+V)\geq Ul'$. Since $P$, $Q$, and
$A$ are fixed and $U$ increases indefinitely as the length of $I_k$
increases choosing a large enough $k$ results in a contradiction.

(2) Find $k$ large enough so that one can find small enough $l$ and
$l+l'$, where $l'>2$ that are both solutions of $a+lA+\omega_u\lfloor
n_{3-u}(k)/|\omega_{3-u}|\rfloor+r=a+\omega_u\lfloor
n_{3-u}(j)/|\omega_{3-u}|\rfloor+r+mD$ for some $m$. That such
solutions exist follows from $\omega_u\lfloor
n_{3-u}(k)/|\omega_{3-u}|\rfloor=\omega_u(n_{3-u}(k)-n_{3-u}(k)\bmod|A\omega_{3-u}|)/|\omega_{3-u}|+\omega_u\lfloor
n_{3-u}(k)\bmod|A\omega_{3-u}|/|\omega_{3-u}|\rfloor=Am'+\omega_u\lfloor
c/|\omega_{3-u}|\rfloor$ and $\omega_u\lfloor
n_{3-u}(j)/|\omega_{3-u}|\rfloor=Am''+\omega_u\lfloor
c/|\omega_{3-u}|\rfloor$ for some integer $m'$ and $m''$. The rest is
similar to case~(1).

(3) Similar to case~(1) with $n_u(\cdot)$ replacing $n_{3-u}(\cdot)$.
\end{proof}

The following theorem is a corollary of Lemma~\ref{boundt}.
\begin{theorem}\label{tcmired}
If there exists a TCMI that counts in some radix $\ttb\geq2$ then
there exists a TCMI $M$ that counts in $\ttb$ such that
$M\dg{I}(w_I,n^I_1, n^I_2)$ where $n^I_-=0$ for every input $I$.
\end{theorem}
\begin{proof}
Use Lemma~\ref{boundt} to find $K$ and change the finite control of the
original TCMI so that all counter values less than $K$ are kept in the
buffers in the finite control.
\end{proof}

The utility of the theorem above lies in the idea that each
computation performed by a TCMI between inputting two successive
digits is a computation of a new value based on the intermediate
result of the previous computation, i.e.\ only one value is needed to
represent the intermediate result. Hence, if the TCM `responsible' for
each digit can be
simulated by some other machine under the assumption that it starts
with one empty register, the TCMI can now also be simulated by a family of
such machines.

Paper~\cite{Schroep} introduced the concept of a so-called {\it more
  powerful one register machine\/} or MP1RM. This machine with a
somewhat tongue-in-cheek name has a single register $r$ and allows the
following operations: $r\leftarrow r+K$, $r\leftarrow r\times K$, 
{\bf if $r-1\geq0$ then $r\leftarrow r-1$ else goto $l$}, {\bf
  $\{\,r\leftarrow r\div K${\rm;} goto $S(r\bmod K)\,\}$}, {\bf goto $l$}, {\bf
  halt}. We can now extend the definition of MP1RM by adding `{\bf if
$I=d$ then goto $l$}' commands, analogous to the way 
a TCMI is derived from a TCM. Naturally there are only finitely
many constants $K$ and $S(K)$ allowed.

Thus extended, an MP1RM becomes an
{\it MP1RM with input\/} or MP1RMI. Other definitions (such as counting
in some radix) can be naturally extended to MP1RMIs as
well. R.~Schroeppel shows in \cite{Schroep} that every TCM can be
emulated with an MP1RM.
\begin{theorem}\label{tcexact}
There is no MP1RMI (therefore no TCMI) that counts online in some
radix $\ttb\geq2$.
\end{theorem}
\begin{proof}
Suppose such a machine $M$ exists. Let $B$ be the product of all the
constants mentioned in the instruction set of $M$
(see the description of an MP1RM above) as well as the input radix, $\ttb$. 
First show that there exists a strictly
increasing finite sequence $n_1, n_2,\ldots, n_k$, and a waiting state $q$ of
$M$ such that $M$ enters $q$ after counting to $\sum_{i=1}^kB^{n_i}$
(which can be picked arbitrarily large) as
well as after counting to $\sum_{i=1}^kB^{n_i}+\sum_{i=k+1}^{K}B^{n_i}$ for arbitrarily
large $n_{k+1}$ and some strictly increasing $n_{k+1},\ldots,n_K$. If
there is no such $q$, one can construct, by induction, a finite
increasing sequence $n_1,\ldots,n_T$ such that for any $n>n_T$ and any
waiting state $q$ of $M$, the waiting state $M$ enters after counting
to $\sum_{i=1}^TB^{n_i}+B^n$ is different from $q$. A contradiction. 

Let $q$ be the waiting state mentioned above, let $n_1,n_2,\ldots,n_k$ 
be the appropriate sequence and let $L$ be the number of steps
it takes $M$ to process the next input digit, 1, after it has counted to
$\sum_{i=1}^kB^{n_i}$. Thus, after starting in $q$ and
$\sum_{i=1}^kB^{n_i}$ in its only register, after inputting 1, $M$
will output (in the register) $\sum_{i=1}^kB^{n_i}+\ttb^S$ where $S$ is
the smallest power such that $\ttb^S>\sum_{i=1}^kB^{n_i}$.

Pick $n_{k+1}>\max\{L,n_k,R\}$, where $R$, to
be specified later,
depends only on the computation $M$ performs to input the last digit
of $\sum_{i=1}^kB^{n_i}+\ttb^S$, such that $M$
enters $q$ after counting to $\sum_{i=1}^kB^{n_i}+\sum_{i=k+1}^{K}B^{n_i}$
for some $n_{k+2},\ldots,n_K$. Note that the only (meaningful) 
branching can be performed by the various $\ldots\div K$ instructions. By
arranging $\sum_{i=1}^kB^{n_i}$ to be large enough we can assume that the
sole register of $M$ is never equal 0 during the processing of the
last input digit of
$\sum_{i=1}^kB^{n_i}+\ttb^S$. It is easy to see that after counting to
$\sum_{i=1}^kB^{n_i}+\sum_{i=k+1}^{K}B^{n_i}$ and inputting 1, $M$ must
stop after the same $L$ steps as it did after counting to
$\sum_{i=1}^kB^{n_i}$ and inputting 1. Indeed, $M$ is in the same
waiting state $q$ before inputting 1 in both cases. Now, our choice of $B$ and
$n_{k+1}>L$ imply that the remainders of the $\ldots\div K$
instructions are unaffected by the addition of $\sum_{i=k+1}^{K}B^{n_i}$.

Therefore, after counting to $\sum_{i=1}^kB^{n_i}+\sum_{i=k+1}^{K}B^{n_i}$
and inputting 1, $M$ will count to
$\sum_{i=1}^kB^{n_i}+\ttb^S+(P/Q)\sum_{i=k+1}^{K}B^{n_i}$ for some integer
$P$ and $Q$ such that $Q$ divides $\sum_{i=k+1}^{K}B^{n_i}$. Note that
both $P$ and $Q$ are determined by the computation $M$ performs after
counting to $\sum_{i=1}^kB^{n_i}$ and inputting 1. Thus choosing
$R$ large enough, it can be arranged that
$\sum_{i=1}^kB^{n_i}+\ttb^S+(P/Q)\sum_{i=k+1}^{K}B^{n_i}\not=
\sum_{i=1}^kB^{n_i}+\sum_{i=k+1}^{K}B^{n_i}+\ttb^{S'}$, where $S'$ is the
smallest such that $\ttb^{S'}>
\sum_{i=1}^kB^{n_i}+\sum_{i=k+1}^{K}B^{n_i}$. One way to see this is by
noting that if $n_{k+1}$ is large enough so that the representation of
$(1/Q)B^{n_{k+1}}$ in radix $\ttb$ has more than $S$ initial zeroes then the
required property holds. 
\end{proof}

The theorem above can be extended in a number of ways. With a slightly
more complicated technique one can show that if the intermediate
values differ from the online values by a bounded constant, the same
result holds, i.e.~such a machine cannot count in $\ttb$, or if the
TCM stage that processes $\gstop$ has a bounded number of {\it
  reversals\/} (or stages as they are called above), the result still
holds. If, on the other hand, there are no other restrictions on this
post-processing stage, new ideas are needed to establish the
nonexistence of such machines. Note that computations
with MP1RMs can be very (potentially infinitely) long. In fact, the history
of the famous {\it Collatz conjecture\/} or the $3n+1$ problem (see
\cite{Leht} for a discussion and a bibliography) shows that
even analyzing the length of the computation performed by the simplest
MP1RMs is very hard.

The last simple result shows that in order for TCMIs to be able to
perform computations of the same level of sophistication as TCMs
when the input is presented starting with the least significant digit,
they need to be able to count in the appropriate radix. This can be
viewed as a `factorization' theorem where a computation on the input
is split into a decoding stage followed by the actual computation.
\begin{theorem}\label{mmach}
Suppose, for a total function $f$ with an unbounded range there exists a
TCM $M_f$ that computes $f$. If $\ttb$ is some radix and there is a TCMI $M$ such that
$M\dg{I\gstop}(w,f(\ulcorner I\urcorner),0)$ for
every string of $\ttb$-digits then there
exists a TCMI $M'$ that counts in $\ttb$. 
\end{theorem}
\begin{proof}
Using Theorem~3.4 of \cite{ibarra} find $a$, $b$, and $c$ such that
$f(a+bn)=f(a)+cn$ for all $n$. Note that it follows from the proof of
\cite{ibarra}, Theorem~3.4 that $c>0$. It is not very difficult to create a finite
control that for any input $I$ holds $2m$ most significant digits of the expression
$a+b\ulcorner I'\urcorner$ for any initial substring $I'$ of $I$ and
correctly updates them upon the input of 
the next digit $d$ of $I$. The constant $m$ is chosen large enough so
that all the relevant constants ($a$, $b$, $c$, $\ttb$ and some of
their products and sums) can be represented in under $m$
$\ttb$-digits. Note that the $m$ least significant digits held by the
finite control are the actual digits of the value of $a+b\ulcorner
I\urcorner$. Now the TCMI $M'$ will utilize the finite control just
built to feed the digits of $a+b\ulcorner I\urcorner$ to the input of
$M$ with some delay and wait for $M$ to compute $f(a+b\ulcorner
I\urcorner)=f(a)+c\ulcorner I\urcorner$. $M'$ then subtracts $f(a)$
and divides by $c$ to compute $\ulcorner I\urcorner$.
\end{proof}

Using the $M_f$ in the statement of the previous theorem as a
`back-end' to $M$ to convert between $n$ and $f(n)$ (and bypassing
$M_f$ in the case of a leading ${\tt 0}$ which can be flagged by $M$'s
finite control) one can show that $M$ cannot compute the value of
$f(n)$ online.

\begin{theorem}\label{tcmio}
Suppose, for a total function $f$ with an unbounded range there exists a
TCM $M_f$ that computes $f$. Then for any radix $\ttb$ there is no TCMI $M$ such that
$M\dg{I}(w,f(\ulcorner I\urcorner),0)$ for every string of
$\ttb$-digits.
\end{theorem}
\begin{proof}
Otherwise using the proof of the previous theorem and the remark that
follows it one can construct a TCMI that counts online in $\ttb$
contradicting Theorem~\ref{tcexact}.
\end{proof}

\section{Open Questions}

We conclude this paper with some questions we have been unable to
answer despite our best efforts.

It is easy to see that the languages $P_{m,n}$ defined in
Section~\ref{autom} are context free if and only if they are
regular. It would be surprising if the answer to the question below is
affirmative, however, the authors do not know how to show that it is negative.

\begin{question}
Do there exist incommensurable radices $\ttd$, $\ttD$, a radix $\ttb$, and a precision $n>0$ 
such that all $P_{m,n}$'s for $\ttd^{m-1}\leq m<\ttd^n$ are regular assuming $n\geq 2$ if $\ttd=2$?
\end{question}

The proof of Theorem~\ref{pdacs} establishes that neither $L_1$ nor $M_1$
is context-free. It is not clear at the moment how to show that all
nontrivial $L_d$'s and $M_d$'s fail to be context-free.
\begin{question}
Given a radix $\ttd>2$ is it true that every $L_d$ and $M_d$,
$0<d<\ttd$ is not a CFL?
\end{question}

It would also be interesting to know more about the computational
power of TCMI.

\begin{question}
Does there exist a TCMI that counts in some radix $\ttb$?
\end{question}
\bibliography{pda2}{}
\bibliographystyle{plain}
\end{document}